\documentclass[11pt]{article} % use larger type; default would be 10pt

\usepackage[utf8]{inputenc} % set input encoding (not needed with XeLaTeX)

\usepackage{geometry} % to change the page dimensions
\usepackage{enumerate}
\usepackage{amsthm}

\usepackage{graphicx,color}    % for latex
\usepackage{epsfig}

\newtheorem{theorem}{Theorem}

\newtheorem{lemma}[theorem]{Lemma}

\newtheorem{corollary}[theorem]{Corollary}

\newtheorem{question}{Question}

\usepackage{amsmath}
\usepackage{amssymb}

\usepackage{caption}
\usepackage{subcaption}
\usepackage{tkz-graph}
\usepackage{pifont}
\usetikzlibrary{matrix}
\usetikzlibrary{decorations.markings}
\tikzstyle{vertex}=[circle, draw, fill=black!50,
                        inner sep=0pt, minimum width=4pt]
\tikzset{->-/.style={decoration={
  markings,
  mark=at position .5 with {\arrow{>}}},postaction={decorate}}}

\newcommand\efface[1]{}

\title{The Coloring Game on Planar Graphs with Large Girth, \\ by a result on Sparse Cactuses}
\author{Clément Charpentier\footnote{This research is supported by the ANR project GAG (Games and graphs), ANR-14-CE25-0006, 2015-2018.}
 \mbox{}\\
 {\small Institut Fourier, Université Joseph Fourier, UMR 5582, Grenoble}\\
 {\small Maths à Modeler}
} 

\date{\today} % Activate to display a given date or no date (if empty),
\begin{document}

\maketitle

\begin{abstract}
	We denote by $\chi_g(G)$ the \emph{game chromatic number} of a graph $G$, which is the smallest number of colors Alice needs to win the \emph{coloring game} on $G$. We know from Montassier et al. [M. Montassier, P. Ossona de Mendez, A. Raspaud and X. Zhu, \emph{Decomposing a graph into forests}, J. Graph Theory Ser. B, 102(1):38-52, 2012] and, independantly, from Wang and Zhang, [Y. Wang and Q. Zhang. \emph{Decomposing a planar graph with girth at least 8 into a forest and a matching}, Discrete Maths, 311:844-849, 2011] that planar graphs with girth at least 8 have game chromatic number at most 5. 

	One can ask if this bound of 5 can be improved for a sufficiently large girth. In this paper, we prove that it cannot. More than that, we prove that there are \emph{cactuses} $CT$ (i.e. graphs whose edges only belong to at most one cycle each) having $\chi_g(CT) = 5$ despite having arbitrary large girth, and even arbitrary large distance between its cycles.
\end{abstract}

\section{Introduction}

We only consider in this paper simple, finite, and undirected graphs. The \emph{length} of a path or cycle is the cardinal of its edge-set. The \emph{girth} $g(G)$ of a graph $G$ is the length of its smallest cycle. A \emph{cactus} is a graph $G$ in which any edge belongs to at most one cycle. The \emph{cycle-distance} of a cactus is the length of its smallest path between two vertices belonging to different cycles. 
For a vertex $v$, we call \emph{$v$-leaf} a vertex of degree 1 (or \emph{leaf}) whose neighbor is $v$.  

The \textit{coloring game} on a graph $G$ is a two-player non-cooperative game on the vertices of $G$, introduced by Brams \cite{Brams1981} and rediscovered ten years after by Bodlaender \cite{Bodlaender1991}. Given a set of $k$ colors, Alice and Bob take turns coloring properly an uncolored vertex, with aim for Alice to color entirely $G$, and for Bob to prevent Alice from winning. The \textit{game chromatic number} $\chi_g(G)$ of $G$ is the smallest number of colors insuring Alice's victory. 
This graph invariant has been extensively studied these past twenty years, see for example \cite{Guan1999, Raspaud2009, Zhu2000, Zhu2008}. 

\medskip
In \cite{Bodlaender1991}, Bodlaender proved that every forest $F$ has $\chi_g(F) \le 5$, and exhibited trees $T$ with $\chi_g(T) \ge 4$. In \cite{Faigle1993}, Faigle et al. showed that every forest $F$ has $\chi_g(F) \le 4$. Conditions for trees to have game chromatic number 3 were recently studied by Dunn et al. \cite{Dunn2014}. 

A graph is said \emph{$(1,k)$-decomposable} if its edge set can be partitionned into two sets, one inducing a forest and the other inducing a graph with maximum degree at most $k$. 
 Using the notion of \emph{marking game} introduced by Zhu in \cite{Zhu1999}, He et al. observed in \cite{He2002} that every $(1,k)$-decomposable graph has $\chi_g(G) \le k + 4$, then deduced upper bounds for the game chromatic number of planar graphs with given girth. Among other results, they proved that planar graphs with girth at least 11 are $(1,1)$-decomposable, and therefore their game chromatic number is at most 5. Later, were proved successively the $(1,1)$-decomposability of planar graphs with girth 10 by Bassa et al. \cite{FM10}, girth 9 by  Borodin et al. \cite{BorodinFM9}, and girth 8 by Montassier et al. \cite{MontassierTA} and Wang and Zhang \cite{WangZhang2011} independantly. There exist planar graphs with girth 7 that are not $(1,1)$-decomposable. 

Borodin et al. \cite{BoroIKS2009} gave conditions for planar graphs with no small cycles except triangles to be $(1,1)$-decomposable, in terms of distance between the triangles and of minimal length of a non-triangle cycle. In \cite{Sidorowicz2007}, Sidorowicz, arguing that cactuses are $(1,1)$-decomposable, showed that every cactus $CT$ has $\chi_g(CT) \le 5$. Moreover, she exhibited a cactus with game chromatic number 5, depicted in Figure \ref{cactusSido}. As one can see, this cactus has intersecting triangles. 

\medskip

\begin{figure}
\center
\begin{tikzpicture}[scale=0.3]

	\draw(27, -1.5) node[vertex](top1){};
	\draw(27.5, -1.5) node[vertex](top2){};
	\draw(28.5, -1.5) node[vertex](top3){};
	\draw(29, -1.5) node[vertex](top4){};

	\draw (28,0) -- (top1);
	\draw (28,0) -- (top2);
	\draw (28,0) -- (top3);
	\draw (28,0) -- (top4);

\foreach \i in {0, 4, ..., 24}{
	\draw(\i, 0) node[vertex](left){}; 
	\draw(\i+4, 0) node[vertex](right){};
	\draw(\i+2, 2.5) node[vertex](top){};

	\draw(\i+1, 4) node[vertex](top1){};
	\draw(\i+1.5, 4) node[vertex](top2){};
	\draw(\i+2.5, 4) node[vertex](top3){};
	\draw(\i+3, 4) node[vertex](top4){};

	\draw(\i-1, -1.5) node[vertex](bot1){};
	\draw(\i-0.5, -1.5) node[vertex](bot2){};
	\draw(\i+0.5, -1.5) node[vertex](bot3){};
	\draw(\i+1, -1.5) node[vertex](bot4){};

	\draw (left) -- (right);
	\draw (left) -- (top);
	\draw (top) -- (right);

	\draw (left) -- (bot1);
	\draw (left) -- (bot2);
	\draw (left) -- (bot3);
	\draw (left) -- (bot4);

	\draw (top) -- (top1);
	\draw (top) -- (top2);
	\draw (top) -- (top3);
	\draw (top) -- (top4);
}

%	\draw (20, 0) node[vertex] (right4) {};

%	\draw (right3) -- (right4);
\end{tikzpicture}
\caption{A cactus with game chromatic number 5 \cite{Sidorowicz2007}}
\label{cactusSido}
\end{figure}
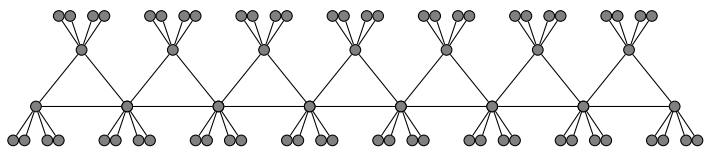

The work we present here started as we tried to answer the following question: 

\begin{question} \label{qfor4}
	Is there an integer $g$ such that every planar graph $G$ with girth at least $g$ has $\chi_g(G) \le 4$ ?
\end{question}

We answer negatively, with a result going way beyond the question we initially asked. 

\begin{theorem}\label{cactus5}
	For any integers $d, k$, there are cactuses $CT$ with girth at least $k$, cycle-distance at least $d$ and $\chi_g(G) = 5$.
\end{theorem}

This proves that the upper bound of 5 for the game chromatic number of the classes of $(1,1)$-decomposable graphs considered in \cite{FM10, BoroIKS2009, BorodinFM9, He2002, MontassierTA, WangZhang2011} are best possible. 

As our construction needs odd cycles, this asks whether or not this result generalizes to bipartite graphs, seeming more difficult to handle. This question is still open. As a partial result, we can find in  \cite{AndresHoch} a proof by Andres and Hochstättler that every \emph{forest with thin 4-cycles}, which is a cactus constructed from a forest by replacing some edges $uv$ by a pair of 2-vertices both adjacent to $u$ and $v$ (and which is bipartite), has game chromatic number at most 4. 

\section{Proof of Theorem \ref{cactus5}}

We consider Alice and Bob playing the coloring game on a graph $G$ with a set of four colors ${\mathcal C} = \{$\ding{182},\ding{183},\ding{184},\ding{185}$\}$. At each time of the game, we denote by $\phi(v)$ the color of a vertex $v$ (if $v$ is colored), and by $\Phi(v)$ the set of colors in the neighborhood of $v$: if $v$ is uncolored, then this is the set of colors forbidden for $v$. We call \emph{surrounded} an uncolored vertex $v$ with $\Phi(v) = {\mathcal C}$. Bob wins if he can surround a vertex. 

\medskip

We give further the construction of $G$.  For now just assume that \emph{every cycle of $G$ is odd} and that \emph{every non-leaf vertex of $G$ is adjacent to a large number of leaves}, say at least 8 leaves. 

We also give further Bob's strategy in details, but assume that \emph{Bob only plays on the leaves of $G$}. Also, for any vertex $v$, when Bob colors a $v$-leaf, he uses a color that is not already in $\Phi(v)$. If Alice colors a $v$-leaf during the game, then Bob always colors another $v$-leaf if possible. So we can assume that \emph{for any uncolored vertex $v$, there is always at least one uncolored leaf of $v$}. 
Moreover, coloring a $v$-leaf for Alice does not suppress any possibility for Bob (this forbid Bob to color $v$ with the color Alice used, but Bob had no intention to color $v$ anyway) and only increase $\Phi(v)$, coloring a leaf is always unoptimal for Alice. So we assume \emph{Alice never colors a leaf during the game}. 

\medskip

 We describe some \emph{winning positions} for Bob (i.e. partial colorings from where Bob has a winning strategy) in the following lemmas. In the description of every winning position, we assume that it is Alice's turn to play. 

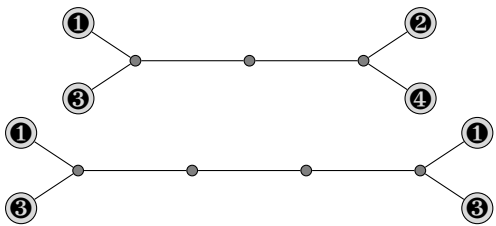
\begin{figure}
\center
\begin{tikzpicture}[scale=0.5]

	\draw(1.5,0) node[vertex, fill=black!15](v1){\ding{182}};
	\draw(3,-1) node[vertex](v2){};
	\draw(6,-1) node[vertex](v3){};
	\draw(9,-1) node[vertex](v4){};
	\draw(10.5,0) node[vertex, fill=black!15](v5){\ding{183}};

	\draw(1.5,-2) node[vertex, fill=black!15](v1p){\ding{184}};
	\draw(10.5,-2) node[vertex, fill=black!15](v5p){\ding{185}};

	\draw (v1p) -- (v2);
	\draw (v1) -- (v2);
	\draw (v2) -- (v3);
	\draw (v3) -- (v4);
	\draw (v4) -- (v5);
	\draw (v4) -- (v5p);
\end{tikzpicture}

\begin{tikzpicture}[scale=0.5]
	\draw(1.5,0) node[vertex, fill=black!15](v1){\ding{182}};
	\draw(3,-1) node[vertex](v2){};
	\draw(6,-1) node[vertex](v3){};
	\draw(9,-1) node[vertex](v4){};
	\draw(12,-1) node[vertex](v5){};
	\draw(13.5,0) node[vertex, fill=black!15](v6){\ding{182}};

	\draw(1.5,-2) node[vertex, fill=black!15](v1p){\ding{184}};
	\draw(13.5,-2) node[vertex, fill=black!15](v6p){\ding{184}};

	\draw (v1p) -- (v2);
	\draw (v1) -- (v2);
	\draw (v2) -- (v3);
	\draw (v3) -- (v4);
	\draw (v4) -- (v5);
	\draw (v5) -- (v6);
	\draw (v5) -- (v6p);
\end{tikzpicture}
\caption{Two winning paths (Lemma \ref{winningpath})}
\label{wpFig}
\end{figure}

\begin{lemma} \label{winningpath}
	Suppose that $G$ contains a path of length $d$, $P = v_0 \ldots v_d$, of uncolored non-leaf vertices. Also suppose $|\Phi(v_0)| \ge 2$, say $\{$\ding{182},\ding{184}$\} \subseteq \Phi(v_0)$. The game is in a winning position for Bob if (see Figure \ref{wpFig}):
\begin{itemize}
	\item $d$ is odd and $\{$\ding{182},\ding{184}$\} \subseteq \Phi(v_d)$.
	\item $d$ is even and $\{$\ding{183},\ding{185}$\} \subseteq \Phi(v_d)$.
\end{itemize}
	We say that $P$ is a \emph{winning path}. 
\end{lemma}

\begin{proof}
Recall that, by assumption, every vertex of $P$ has at least one uncolored leaf. 

In the case $d =0$, $P$ is reduced to a single surrounded vertex and Bob wins. 

The case $d = 1$ corresponds to an edge $v_0v_1$ with $v_0$ and $v_1$ uncolored and $|\Phi(v_0) \cap \Phi(v_1)| \ge 2$, say $\{$\ding{182},\ding{184}$\} \subseteq \Phi(v_0) \cap \Phi(v_1)$. If Alice colors $v_0$, she has to use an even color, Bob colors a $v_1$-leaf with the other even color, surrounds $v_1$, and wins. If Alice colors $v_1$, then Bob can surrounds $v_0$ as well. So Bob's winning strategy consists in coloring leaves of $v_0$ with different colors until $v_0$ is surrounded or Alice colors $v_0$ or $v_1$.

\medskip

We prove the other cases by induction. Assume that our lemma is true for every value of $d$ strictly smaller than an integer $q$ and consider the case $d = q$. Bob's strategy is to colors a $v_0$-leaf until $v_0$ is surrounded or Alice colors a vertex of $P$. When Alice colors a vertex $v_i$ of $P$, we consider two cases.
\begin{itemize}
	\item If $i = 0$ or $i = q$, say w.l.o.g. $i = 0$, then she uses an even color. Bob colors a $v_1$-leaf with the other even color. By induction, Bob wins since $P - v_0$ is a winning path of length $q - 1$.  
	\item If $i \ne 0$ and $i \ne q$, then let $P_1 = v_0 \ldots v_{i-1}$ and $P_2 = v_{i+1} \ldots v_q$. We denote by $d_1$ and $d_2$ the length of $P_1$ and $P_2$ respectively. We have $d_1 + d_2 = q - 2$. We consider two subcases: 
\begin{itemize}
		\item Either $q$ is even, and $d_1$ and $d_2$ are either both even or both odd. Say we have $\{$\ding{182},\ding{184}$\} \subseteq \Phi(v_0)$ and $\{$\ding{183},\ding{185}$\} \subseteq \Phi(v_q)$. Without loss of generality, assume Alice colored $v_i$ with \ding{182}. If $d_1$ and $d_2$ are both odd, then Bob colors a $v_{i-1}$-leaf with \ding{184}. If they are both even, then Bob colors a $v_{i+1}$-leaf with \ding{184}. Path $P_1$ or $P_2$ respectively is a winning path and Bob wins by induction hypothesis. 
		\item Either $q$ is odd, and say $\{$\ding{182},\ding{184}$\} \subseteq \Phi(v_0) \cap \Phi(v_1)$. Among $d_1$ and $d_2$, one is even and one is odd, say $d_1$ is even and $d_2$ is odd. If Alice colored $v_i$ with an odd color, then Bob colors a $v_{i+1}$-leaf with the other odd color. If Alice colored $v_i$ with an even color, then Bob colors a $v_{i-1}$-leaf with the other even color. Path $P_1$ or $P_2$ respectively is a winning path and Bob wins by induction hypothesis. 
\end{itemize}
\end{itemize}
This concludes our proof. 
\end{proof}

\begin{figure}
\center
\begin{tikzpicture}[scale=0.3]

	\draw(0,0) circle (4);
	\draw[rotate=144](4,0) node[vertex](v3){};
	\draw[rotate=216](4,0) node[vertex](v4){};

	\draw[rotate=144](6,-1) node[vertex, fill=black!15](v3s1){\ding{182}}; 
	\draw[rotate=144](6,1) node[vertex, fill=black!15](v3s2){\ding{184}}; 
	\draw[rotate=216](6,0) node[vertex, fill=black!15](v4s){\ding{182}}; 
	
	\draw (v3) -- (v3s1); 
	\draw (v3) -- (v3s2); 
	\draw (v4) -- (v4s); 

\end{tikzpicture}
\caption{A winning cycle (Lemma \ref{winningcycle})}
\label{wcFig}
\end{figure}
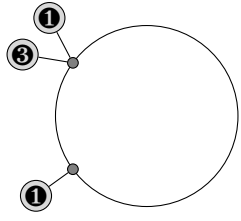

\begin{lemma} \label{winningcycle}
	Suppose that $C$ is an (odd) cycle of $G$ of uncolored vertices. If there are two neighbors $u$ and $v$ with $|\Phi(u)| \ge 2$ and $|\Phi(u) \cap \Phi(v)| = 1$ (see Figure \ref{wcFig}), then $G$ is in a winning position for Bob. We say that $C$ is a \emph{winning cycle}.
\end{lemma}

\begin{proof}
	Let $C = v_0 v_1 \ldots v_k$ be an odd cycle of uncolored vertices with, say, $\{$\ding{182},\ding{183}$\} \subseteq \Phi(v_0)$ and \ding{182} $\in \Phi(v_1)$. If the next Alice's move is not to color $v_0$ or $v_1$, then Bob colors a leaf of $v_1$ with \ding{183} and Bob wins by Lemma \ref{winningpath} ($v_0v_1$ is a winning path). If Alice colors $v_0$, then she has to use a color different to \ding{182} and, at his turn, Bob colors a $v_k$-leaf with \ding{182}. Since $C$ is an odd cycle, Bob wins by Lemma \ref{winningpath} (path $v_1 \ldots v_k$ is a winning path). Similarly, if Alice colors $v_1$, Bob can color a $v_2$-leaf with a color among \ding{182} and \ding{183} different from $\phi(v_1)$ and $v_2  \ldots v_k v_0$ is a winning path. 
\end{proof}

\begin{figure}
\center
\begin{tikzpicture}[scale=0.3]

	\draw(0,0) circle (4);
	\draw[rotate=0](-6,0) node[vertex](v3){};
	\draw[rotate=0](-8,0) node[vertex](v2){};
	\draw[rotate=0](-10,0) node[vertex](v1){};
	\draw[rotate=0](-4,0) node[vertex](v4){};
	\draw[rotate=72](-4,0) node[vertex](v5){};

	\draw(-11, 2) node[vertex, fill=black!15](v3s1){\ding{182}}; 
	\draw(-9,2) node[vertex, fill=black!15](v3s2){\ding{184}}; 
	\draw[rotate=72](-2,0) node[vertex, fill=black!15](v4s){\ding{183}}; 
	
	\draw (v1) -- (v3s1); 
	\draw (v1) -- (v3s2); 
	\draw (v5) -- (v4s); 
	\draw (v1) -- (v2);
	\draw (v2) -- (v3);
	\draw (v3) -- (v4);
\end{tikzpicture}
\caption{A winning cycle-path (Lemma \ref{winningcyclepath})}
\label{wcpFig}
\end{figure}
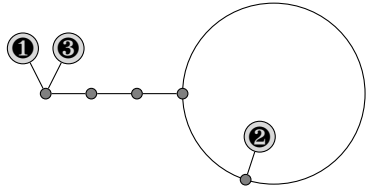

\begin{lemma} \label{winningcyclepath}
	Suppose that $G$ contains a path of length $d$, $P = v_0 \ldots v_d$, of uncolored non-leaf vertices. Also suppose $|\Phi(v_0)| \ge 2$, say $\{$\ding{182},\ding{184}$\} \subseteq \Phi(v_0)$, and suppose $v_d$ belongs to an odd cycle $C$ of uncolored vertices, with $w$ being a neighbor of $v_d$ in this cycle. Graph $G$ is in a winning position for Bob if (see Figure \ref{wcpFig}):
\begin{itemize}
	\item $d$ is even and \ding{182} or \ding{184} is in $\Phi(w)$.
	\item $d$ is odd and \ding{183} or \ding{185} is in $\Phi(w)$.
\end{itemize}
	We say that $P \cup C$ is a \emph{winning cycle-path}. 
\end{lemma}
\begin{proof}
The case $d = 0$ is true by Lemma \ref{winningcycle}, since $C$ is then a winning cycle. 

	We prove the other cases by induction on $d$. 
Assume our lemma is true for every value of $d$ strictly smaller than $q$. We consider the case $d = q$. We have different cases depending of Alice's move: 
\begin{itemize}
	\item If Alice colors $w$, then let $w'$ be the neighbor of $w$ different from $v_q$. The uncolored vertices of $P \cup C$ form a path $P' = v_0, \ldots, v_q, \ldots, w'$. Since $C$ has odd length, one path among $P$ and $P'$ is odd and the other is even. If $\phi(w) \in \Phi(v_0)$, then Bob can color a leaf of $v_q$ or $w'$ such that the odd path among $P$ and $P'$ become a winning path. If $\phi(w) \not \in \Phi(v_0)$, then Bob colors a leaf of $v_q$ or $w'$ such that the even path among $P$ and $P'$ is a winning path. By Lemma \ref{winningpath}, Bob is winning the game. 
	\item If Alice colors $v_0$, then Alice has to use an even color, Bob colors a $v_1$-leaf with the other even color and $(P - v_0) \cup C$ is a winning cycle-path, Bob wins by induction hypothesis. 
	\item If Alice colors $v_q$, then let $w''$ be the other neighbor of $v_d$ in $C$ than $w$. If there is in $\Phi(w)$ a color $c$ different to $\phi(v_q)$, then Bob colors $w''$ with $c$ and $C - v_d$ is winning path since $C$ has odd length. By  Lemma \ref{winningpath}, Bob wins. If $\Phi(w) = \phi(v_d)$, then we have two cases: 
\begin{itemize}
	\item If $d$ is even, then $\phi(v_d) =$ \ding{182} or \ding{184}, say \ding{182}. Bob colors a $v_{d-1}$-leaf with \ding{184} and $P - v_d$ is a winning path. 
	\item If $d$ is odd, then $\phi(v_d) =$ \ding{183} or \ding{185}, say \ding{183}. Bob colors a $v_{d-1}$-leaf with \ding{185} and $P - v_d$ is a winning path. 
\end{itemize}
	Bob wins by Lemma \ref{winningpath}. 
	\item If Alice colors $v_i$, $i \ne 0$ and $i \ne q$, then let $P_1 = v_0 \ldots v_{i-1}$ and $P_2 = v_{i+1} \ldots v_q$. We denote by $d_1$ and $d_2$ the length of $P_1$ and $P_2$ respectively. We have $d_1 + d_2 = q - 2$. We consider two subcases: 
\begin{itemize}
	\item If $d_1$ is odd and Alice used an odd color on $v_i$, or if $d_1$ is even and Alice used an even color on $v_i$, then Bob colors $v_{i-1}$ with the other odd color (resp. the other even color), and Bob wins by Lemma \ref{winningpath} ($P_1$ is a winning path). 
	\item If $d_1$ is odd and Alice used an even color on $v_i$, or if $d_1$ is even and Alice used an odd color on $v_i$, then Bob colors $v_{i+1}$ with the other even color (resp. the other odd color). Bob wins by induction hypothesis, since $P_2 \cup C$ is a winning cycle-path of smaller path-length. 
\end{itemize}
	\item If Alice colors another vertex, then one can observe that Bob can color a $w$-leaf in such a way that $P + w$ is a winning path.
\end{itemize}
Our lemma is true for every $d$ by induction. This concludes our proof. 
\end{proof}

\begin{corollary} \label{obs-cp}
	If, at Bob's turn, there is in $G$ an uncolored path $P = v_0, \ldots, v_d$ such that $|\Phi(v_0)| \ge 2$ and $v_d$ belongs to an uncolored odd cycle, then Bob has a winning strategy (he can obtain a winning cycle-path in one move). 
\end{corollary}

\begin{figure}
\center
\begin{tikzpicture}[scale=0.8]

\foreach \c in {-7.2, -5.6, ..., 7.2}{
	\draw(\c,-1.6) circle (0.6);
}

	\draw[rotate=0](0,0) node[vertex](v3){};
	\draw[rotate=0](-2,0) node[vertex](v2){};
	\draw[rotate=0](-4,0) node[vertex](v1){};
	\draw[rotate=0](2,0) node[vertex](v4){};
	\draw[rotate=0](4,0) node[vertex](v5){};

	\draw[rotate=0](7.2,-1) node[vertex](v5s1){};
	\draw[rotate=0](5.6,-1) node[vertex](v5s2){};
	\draw[rotate=0](4,-1) node[vertex](v4s1){};
	\draw[rotate=0](2.4,-1) node[vertex](v4s2){};
	\draw[rotate=0](-7.2,-1) node[vertex](v1s1){};
	\draw[rotate=0](-5.6,-1) node[vertex](v1s2){};
	\draw[rotate=0](-4,-1) node[vertex](v2s1){};
	\draw[rotate=0](-2.4,-1) node[vertex](v2s2){};
	\draw[rotate=0](0.8,-1) node[vertex](v3s2){};
	\draw[rotate=0](-0.8,-1) node[vertex](v3s1){};

	\draw(-4,0.5) node{$x$}; 
	\draw(-2,0.5) node{$y$}; 
	\draw(0,0.5) node{$z$}; 
	\draw(2,0.5) node{$x'$}; 
	\draw(4,0.5) node{$y'$}; 

	\draw(-7.2,-1.6) node{$C_{x1}$}; 
	\draw(-5.6,-1.6) node{$C_{x2}$}; 
	\draw(-4,-1.6) node{$C_{y1}$};
	\draw(-2.4,-1.6) node{$\ldots$};  

	\draw (v5) -- (v4); 
	\draw (v1) -- (v2);
	\draw (v2) -- (v3);
	\draw (v3) -- (v4);

	\draw (v5) -- (v5s1);
	\draw (v5) -- (v5s2);
	\draw (v4) -- (v4s1);
	\draw (v4) -- (v4s2);
	\draw (v3) -- (v3s1);
	\draw (v3) -- (v3s2);
	\draw (v2) -- (v2s1);
	\draw (v2) -- (v2s2);
	\draw (v1) -- (v1s1);
	\draw (v1) -- (v1s2);

\end{tikzpicture}
\caption{A sketch of graph $G$}
\label{badAlice}
\end{figure}
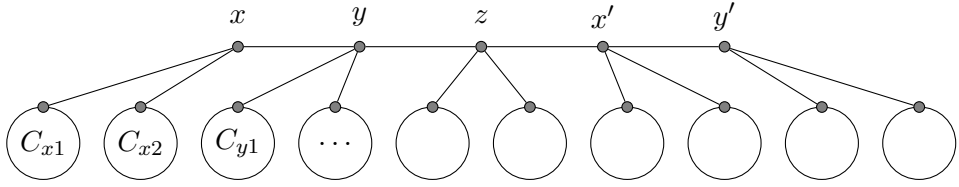

We describe now explicitely the graph $G$ on which Alice and Bob are playing, depicted in Figure \ref{badAlice}. Here paths and cycles have arbitrary large length and cycles are odd. For a vertex $v$, we say we \emph{attach a cycle $C$ to $v$ with a path $P$} if we add $C$ and $P$ such that the endvertices of $P$ are $v$ and a vertex of $C$. We start from a path $xyzy'x'$, then for every vertex $v$ in this path, we attach two odd cycles $C_{v1}$ and $C_{v_2}$ to $v$ with two paths $P_{v1}$ and $P_{v2}$ respectively. We denote $G_v = C_{v1} + C_{v2} + P_{v1} + P_{v2}$.
Finally, we copy our graph (i.e. we add a similar second connected component to the graph). 
Then we add a large number of leaves (at least 8) to every vertex to obtain $G$. 

\medskip

Now we give the winning strategy for Bob. As we assumed, Bob only plays on leaves. After Alice's first move, Bob considers the connected component of the graph where Alice did not play on and will only play on it. 

\begin{description}
\item[{\bf Step 1.}] \emph{Bob plays in $G_z$. He aims to end Step 1 with his victory or with Alice coloring $z$. Moreover, Bob wants to be sure Alice colors in Step 1 at most one vertex not in $G_z$}.  

At his first move, Bob colors a $z$-leaf with an arbitrary color. If Alice colors $z$ at her second move, then Bob goes to Step 2. Otherwise, Bob colors another $z$-leaf with another color. If Alice colors $z$, then Alice played at most one move out of $G_z$ and Bob goes to Step 2. Otherwise, 
\begin{itemize}
	\item If Alice has played her two moves in $G_z$, then Bob colors another $z$-leaf and Alice has to color $z$ for Bob not to surround it next turn. Bob goes to Step 2. 
	\item Otherwise, then we can assume without loss of generality that $P_{z1} \cup C_{z1}$ is uncolored, and Bob wins by Corollary \ref{obs-cp}.  
\end{itemize}

\item[{\bf Step 2.}] Since Alice only played once out of $G_z$, we assume that $G_x \cup G_y$ is uncolored. \emph{Now Bob plays in $G_x$. He aims to end Step 2 with his victory or with Alice coloring $x$ with a color different from $\phi(z)$. Moreover, Bob wants to be sure Alice colors in Step 2 at most one vertex not in $G_x$, and this vertex, if it exists, is not $y$}.  

 Step 2 is quite similar to Step 1. Bob begins by coloring a $x$-leaf with $\phi(z)$. If Alice colors $y$, she has to use a color different from $\phi(z)$ and Bob wins by Corollary \ref{obs-cp}. If she colors $x$, then Bob goes to Step 3. If she colors any other vertex, then Bob colors another $x$-leaf, and, from this point, this  is similar to Step 1. 

\item[{\bf Step 3.}] We can assume without loss of generality that $P_{y1} \cup C_{y1}$ is uncolored. Since $|\Phi(y)| \ge 2$, Bob wins by Corollary \ref{obs-cp}. 
\end{description}

This ends the proof of Theorem \ref{cactus5}.

\bibliographystyle{plain}
\bibliography{gcn-sparse} 

\end{document}